\documentclass[11pt]{amsart}
\usepackage{fullpage}
\usepackage[foot]{amsaddr}
\usepackage{amsmath}
\usepackage{amssymb}
\usepackage{amsthm}

\makeatletter
\@ifundefined{newcounteralias}{\usepackage{aliascnt}}{}
\makeatother

\usepackage[linesnumbered,boxed,ruled,vlined]{algorithm2e}
\usepackage{algpseudocode}
\usepackage{enumitem}
\usepackage{xifthen}
\usepackage{xspace}

\usepackage[margin=1cm]{caption} 
\usepackage{subfig}

\usepackage[thinlines]{easytable}

\usepackage[bookmarks=true,hypertexnames=false,pagebackref]{hyperref}
\hypersetup{colorlinks=true, citecolor=blue, linkcolor=red, urlcolor=blue}

\usepackage{pgfplots}
\pgfplotsset{compat=1.16}
\usepackage{tikz}
\usetikzlibrary{arrows,arrows.meta,backgrounds,calc,fit,decorations.pathreplacing,decorations.markings,shapes.geometric}

\tikzstyle{internal} = [draw, fill, shape=circle]
\tikzstyle{external} = [shape=circle]
\tikzstyle{square}   = [draw, fill, rectangle]
\tikzstyle{triangle} = [draw, fill, regular polygon, regular polygon sides=3, inner sep=3pt]
\tikzstyle{pentagon} = [draw, fill, regular polygon, regular polygon sides=5, inner sep=2pt, minimum size=14pt]
\tikzset{every fit/.append style=text badly centered}

\usetikzlibrary{positioning,chains,fit,shapes,calc}
\usetikzlibrary{trees}
\usetikzlibrary{decorations.pathreplacing}
\usetikzlibrary{decorations.pathmorphing}
\usetikzlibrary{decorations.markings}
\tikzset{>=latex} 

\usepackage{ifthen}

\usepackage{cleveref}

\usepackage[textsize=tiny]{todonotes}

\usepackage[normalem]{ulem}

\usepackage{mleftright}

\usepackage{cool}
\Style{DSymb={\mathrm d},DShorten=true,IntegrateDifferentialDSymb=\mathrm{d}}

\newcommand{\tp}[1]{{\left( #1 \right)}}
\newcommand{\sqtp}[1]{{\left[ #1 \right]}}

\newcommand{\Ex}{\mathop{\mathbb{{}E}}\nolimits}
\renewcommand{\Pr}{\mathop{\mathrm{Pr}}\nolimits}

\def\*#1{\mathbf{#1}}
\def\+#1{\mathcal{#1}}
\def\-#1{\mathrm{#1}}
\def\=#1{\mathbb{#1}}
\def\^#1{\mathbb{#1}}

\newcommand{\norm}[2]{\ensuremath{\left\Vert #2 \right\Vert_{#1}}}
\newcommand{\abs}[1]{\ensuremath{\left\vert#1\right\vert}}
\newcommand{\inner}[2]{\ensuremath{\left\langle #2 \right\rangle}_{#1}}

\newcommand{\set}[1]{\left\{#1\right\}}
\newcommand{\eps}{\varepsilon}

\newcommand{\Var}[2]{\ensuremath{\textnormal{Var}_{#1}\left(#2\right)}}
\newcommand{\mixingtime}[1]{\ensuremath{t_{\textnormal{mix}}(#1)}}
\newcommand{\pimin}{\ensuremath{\pi_{\textnormal{min}}}}

\newcommand{\defeq}{:=}

\newcommand{\numP}{\#{\textnormal{\textbf{P}}}}
\newcommand{\NP}{\textnormal{\textbf{NP}}}
\newcommand{\RP}{\textnormal{\textbf{RP}}}

\newcommand{\DTV}[2]{\-D_{\mathrm{TV}}\left({#1},{#2}\right)}

\newtheorem{theorem}{Theorem}

\makeatletter
\@ifundefined{newcounteralias}{%
  \newcommand{\newsharedtheorem}[2]{%
    \newaliascnt{#1}{theorem}%
    \newtheorem{#1}[#1]{#2}%
    \aliascntresetthe{#1}%
  }%
}{%
  \newcommand{\newsharedtheorem}[2]{\newtheorem{#1}[theorem]{#2}}%
}
\makeatother

\newsharedtheorem{conjecture}{Conjecture}
\newsharedtheorem{lemma}{Lemma}
\newsharedtheorem{claim}{Claim}
\newsharedtheorem{observation}{Observation}
\newsharedtheorem{proposition}{Proposition}
\newsharedtheorem{corollary}{Corollary}
\theoremstyle{definition}
\newsharedtheorem{condition}{Condition}
\newsharedtheorem{definition}{Definition}

\theoremstyle{remark}
\newtheorem{remark}{Remark}

\crefname{theorem}{Theorem}{Theorems}
\crefname{conjecture}{Conjecture}{Conjectures}
\crefname{observation}{Observation}{Observations}
\crefname{claim}{Claim}{Claims}
\crefname{condition}{Condition}{Conditions}
\crefname{algorithm}{Algorithm}{Algorithms}
\crefname{property}{Property}{Properties}
\crefname{example}{Example}{Examples}
\crefname{fact}{Fact}{Facts}
\crefname{lemma}{Lemma}{Lemmas}
\crefname{corollary}{Corollary}{Corollaries}
\crefname{definition}{Definition}{Definitions}
\crefname{remark}{Remark}{Remarks}
\crefname{proposition}{Proposition}{Propositions}
\crefname{equation}{equation}{equations}
\crefname{enumi}{Case}{Case}
\creflabelformat{enumi}{(#2#1#3)}

\makeatletter
\def\prob#1#2#3{\goodbreak\begin{list}{}{\labelwidth\z@ \itemindent-\leftmargin
      \itemsep\z@  \topsep6\p@\@plus6\p@
      \let\makelabel\descriptionlabel}
  \item[\textbf{Name}]#1
  \item[\textbf{Instance}]#2
  \item[\textbf{Output}]#3
  \end{list}}
\makeatother

\makeatletter
\providecommand\@dotsep{5}
\def\listtodoname{Todo list}
\def\listoftodos{\@starttoc{tdo}\listtodoname}
\makeatother

\newcommand{\ZHC}{Z_{\text{HC}}}

\usepackage{nicefrac,comment}

\title{Approximating spin systems on planar graphs}
\author{Heng Guo}
\author{Xinyuan Zhang}
\address[Heng Guo]{School of Informatics, University of Edinburgh, Informatics Forum, Edinburgh, EH8 9AB, United Kingdom}
\address[Xinyuan Zhang]{State Key Laboratory for Novel Software Technology, New Cornerstone Science Laboratory, Nanjing University, 163 Xianlin Avenue, Nanjing, China}
\begin{document}

\begin{abstract}
  We show that the hard-core partition function admits a fully polynomial-time randomised approximation scheme (FPRAS) on planar graphs when the activity is a sufficiently small constant.
  In contrast, we show that for any constant $q\ge 4$, approximately counting $q$-colourings in planar graphs is \NP{}-hard.
  We also give a complete characterisation of when an FPRAS exists for a sufficiently small external field for 2-spin systems on planar graphs.
  The main ideas of all proofs were found using GPT-5.6 Sol Ultra.
\end{abstract}

\maketitle

\section{Introduction}

Spin systems are models for nearest-neighbour interactions.
Originating in statistical physics, they are now studied across a number of areas including theoretical computer science.
The main computational task is to approximate their so-called partition function, or equivalently to sample from the associated Gibbs distribution.
A striking phenomenon is the computational phase transition.
It means that the computational complexity of these approximation tasks can change drastically, 
when the parameter of the system crosses certain critical thresholds.

The first model where this is rigorously established is the hard-core gas model.
Given a graph $G=(V,E)$, let $\+I(G)$ denote the set of independent sets of $G$.
Then, the hard-core partition function is the following:
\begin{align}
  \ZHC{}(G,\lambda)\defeq\sum_{I\in\+I(G)} \lambda^{\abs{I}},
\end{align}
where $\lambda>0$ is a parameter called the fugacity.
Suppose the graph has maximum degree $\Delta$.
Then, there is a critical threshold $\lambda_c=\lambda_c(\Delta)\approx\frac{e}{\Delta}$.
If $\lambda<\lambda_c$, then efficient algorithms exist to approximate $\ZHC$ \cite{Wei06,ALO24},
whereas if $\lambda>\lambda_c$, the problem becomes \NP{}-hard \cite{SS14,GSV16}.
More recently, even the case of $\lambda=\lambda_c$ is shown to admit an efficient algorithm \cite{CCYZ25}.

Planar graphs are an important class of graphs that has attracted particular attention in statistical physics.
An early success is the polynomial-time exact algorithm to compute the Ising partition function and the number of perfect matchings in planar graphs,
namely the FKT algorithm \cite{TF61,Kas61,Kas67}.
It turns out that this is more of an exception, and most partition functions on planar graphs remain \numP{}-hard to compute exactly \cite{GW20,CFS21,CFGW22,CFST26},
including all the problems we are going to study next in this paper.

On the approximation side, since planar graphs are always sparse, it is natural to ask if computational phase transitions similar to the bounded degree case hold in planar graphs as well.
However, planar graphs can embed trees of arbitrarily large maximum degree.
Thus, traditional algorithms do not apply.

Our first result is that, nonetheless, planar graphs admit efficient algorithms for the hard-core partition function when $\lambda$ is a sufficiently small constant.
\begin{theorem}  \label{thm:HC}
  There is an absolute constant $\lambda_0>0$ such that for any $\lambda<\lambda_0$, the following holds.
  There is an algorithm that takes as input a planar graph $G=(V,E)$ and $\eps>0$, and outputs a value $\widetilde Z$ such that $e^{-\eps}\le \frac{\widetilde Z}{Z}\le e^{\eps}$ with probability at least $3/4$.
  It runs in time polynomial in $n$ and $1/\eps$.
\end{theorem}
In particular, taking $\lambda_0=\frac{1}{1443}$ suffices.
We did not optimise the constant $\lambda_0$ or the running time in favour of a simple presentation.
An algorithm satisfying the description in \Cref{thm:HC} is often called a fully polynomial-time randomised approximation scheme (FPRAS), and the probability $3/4$ is not important and can be easily amplified.

Previously, FPRASes are known through Glauber dynamics when the planar graph has maximum degree $\Delta$ and $\lambda=O(1/\sqrt{\Delta})$ \cite{Hay06,Eft26}.
In contrast, \Cref{thm:HC} does not assume any maximum degree requirement.
On the intractability side, Goldberg, Jerrum, and McQuillan \cite{GJM15} showed that for $\lambda>312$, the problem becomes \NP{}-hard, even if the maximum degree is $4$.
Together with \Cref{thm:HC}, this strongly suggests that there indeed is a computational phase transition for approximating the hard-core partition function on planar graphs.

In light of \Cref{thm:HC}, a natural guess is then the same phenomenon should happen for any spin system that undergoes a computational phase transition for bounded degree graphs.
For example, approximately counting proper colourings in graphs with maximum degree $\Delta$ is shown to be \NP{}-hard if $q$ is even and $q<\Delta$ \cite{GSV15},
and FPRASes exist for $q>1.809\Delta$ and $\Delta$ large enough \cite{CV25}, or $q>(11/6-\eps_0)\Delta$ for a small constant $\eps_0>0$ and any $\Delta\ge 3$ \cite{CDMPP19}.
Thus, one may guess that an FPRAS exists when the number of colours is sufficiently large for planar graphs.
This guess is further strengthened by a result of \cite{HVV15},
which states that for planar graphs with maximum degree $\Delta$, an FPRAS exists for $q=O(\Delta/\log\Delta)$, below the general graph hardness threshold.
However, our second result refutes this possibility.
\begin{theorem}  \label{thm:colouring}
  For any $q\ge 4$,
  if there is an FPRAS for the number of $q$-colourings in planar graphs,
  then \NP{}=\RP{}.
\end{theorem}
Deciding the existence of $3$-colourings in planar graphs is \NP{}-hard \cite{GJS76}.
Thus, the only tractable version for planar graphs is to count $2$-colourings.

\Cref{thm:colouring} is somewhat surprising, especially given the four colour theorem \cite{RSST97},
which states that a $4$-colouring always exists and can be found in polynomial-time for planar graphs.
On the other hand, Welsh \cite[Conjecture 8.7.5]{Wel93} conjectured the \NP{}-hardness for all $q\ge 4$, although he admitted that the conjecture may be rash.
\Cref{thm:colouring} confirms this rash conjecture.
The number of $q$-colourings is also an evaluation of the chromatic polynomial, which is a special case of the Tutte polynomial.
Thus, \Cref{thm:colouring} also resolves the open cases in \cite{GJ12}, which studies the approximation complexity of evaluating the Tutte polynomial in planar graphs.
Due to planar duality, the hardness in \Cref{thm:colouring} carries over to the flow polynomial, another specialisation of the Tutte polynomial, as well.

In view of \Cref{thm:HC,thm:colouring}, it is natural to wonder about the general principle for computational phase transition to happen in planar graphs.
We do not have a complete answer, but we worked out the case of 2-spin systems, whose partition function is defined as follows for a graph $G=(V,E)$:
\begin{align*}
    Z_{(\beta,\gamma,\lambda)}(G) \defeq \sum_{\sigma:V\to\{0,1\}} \prod_{(u,v)\in E} \beta^{\mathbf{1}\{\sigma(u)=\sigma(v)=0\}} \gamma^{\mathbf{1}\{\sigma(u)=\sigma(v)=1\}} \prod_{v\in V} \lambda^{\mathbf{1}\{\sigma(v)=1\}},
\end{align*}
where $\beta,\gamma,\lambda$ are three parameters.
The hard-core model is a special case with $\beta=1$, $\gamma=0$, and $\lambda>0$,
and the famous Ising model is a special case with $\beta=\gamma>0$.
On planar graphs and when $\lambda=1$, the Ising partition function can be computed exactly in polynomial time via the FKT algorithm \cite{TF61,Kas61,Kas67}.
Other than this exceptional case, due to symmetry, we may assume that $\beta\ge\gamma$, and we have the following result.

\begin{theorem}  \label{thm:2-spin}
  Let $\beta>0$, $\gamma\ge 0$, $\lambda>0$ be three parameters such that either $\gamma<\beta$ or $\gamma=\beta$ and $\lambda\neq 1$.
  Then, on planar graphs,
  \begin{enumerate}
    \item \label{item:2-spin-hard} if $\beta<1$, then for any $\lambda>0$, there is no FPRAS for $Z_{(\beta,\gamma,\lambda)}$ unless \NP{}=\RP{};
    \item \label{item:2-spin-easy} otherwise, there always exists an absolute constant $\lambda_0>0$ such that an FPRAS exists for $Z_{(\beta,\gamma,\lambda)}$ whenever $\lambda<\lambda_0$.
  \end{enumerate}
\end{theorem}
Again, the constant in \Cref{thm:2-spin} can be taken as $\lambda_0=\frac{1}{1443}$.

In fact, when $\beta\gamma=1$, then the partition function factorises and can be computed exactly in polynomial time.
In the so-called ferromagnetic regime $\beta\gamma>1$, there is always $\lambda_0>0$ such that an FPRAS exists for $\lambda<\lambda_0$ for general graphs \cite{GJP03}.
As a side note, the approximation complexity of partition functions of ferromagnetic 2-spin systems in general graphs is still not completely clear, despite a lot of effort \cite{GJP03,LLZ14,GL18,GLL20,FGY26a,FGY26b}.

The main contribution of \Cref{thm:2-spin} lies in the antiferromagnetic regime $\beta\gamma<1$.
In that case, combined with the hardness result in \cite{GJM15}, \Cref{thm:2-spin} implies that computational phase transition can happen in planar graphs if and only if $0\le \gamma<1\le\beta$.
In some sense, this is because pinnings in this regime only make the effective $\lambda$ smaller, and thus the parameters stay in the same regime.
On the other hand, if $\gamma,\beta<1$, then a vertex can have a lot of neighbours pinned to $1$, which makes the effective $\lambda$ arbitrarily large, preventing the algorithm from generalising to this case.
In fact, the same barrier exists for planar colourings as well, namely, the effective number of colours can be arbitrarily small with a bad pinning.
We expect all spin systems that can stay within the same parameter regime after arbitrary pinnings to admit a computational phase transition in planar graphs.

\subsection{Technical overview}

The main ideas of all proofs in this paper were found by GPT-5.6 Sol Ultra.
For consistency with standard mathematical exposition, the words ``we'' and ``our'' are used throughout the paper, including when presenting ideas that originate in the output of the model.
The authors simplified, streamlined, and wrote all of the proofs. 
The authors take full responsibility for the correctness of the paper.

For \Cref{thm:HC}, the main idea is to use block dynamics, where the blocks are given by a star colouring \cite{ACKKR04} using $20$ colours.
A star colouring is a proper colouring and any two colour classes form a star forest.
For any planar graph, such a colouring always exists and can be found in polynomial time.
We develop a framework where bounding the Hirschfeld–Gebelein–Rényi maximal correlation \cite{Ren59} between two blocks under arbitrary pinnings is sufficient to bound the spectral gap of the block dynamics.
The Hirschfeld–Gebelein–Rényi maximal correlation is particularly convenient for our purpose, since it tensorises across the stars.
Via a standard comparison argument, the single-site Glauber dynamics also mixes rapidly.
See \Cref{sec:Glauber}.

After uploading our paper to arXiv, Zongchen Chen pointed out to us that the key quantity we need to bound coincides with the pairwise spectral influence introduced by Leake and Oveis Gharan \cite{LO25}.
Our \Cref{lem:dobrushin} actually resolves an open problem of theirs.
See \Cref{rem:LO25}.

For the hardness proofs, the ideas are rather standard.
We find a suitable optimisation problem so that we can construct a new graph, where the dominating contribution to the partition function comes from configurations of the new graph corresponding to the optimal solution of the original instance.
The tricky part though is to find the right problem to reduce from, and this is where GPT-5.6 Sol Ultra is particularly helpful.
For counting colourings in planar graphs, we reduce from a problem called graph coalition structure generation (GCSG), which is known to be \NP{}-hard for planar graphs \cite{VPJ12}.
The hardness here essentially comes from maximum independent set in planar graphs \cite{GJS76}, but the form of GCSG is a lot more convenient.
For anti-ferromagnetic 2-spin systems, instead, we give a direct reduction from the maximum independent set problem in planar graphs.

\section{FPRAS for planar hard-core}

To approximate $\ZHC$, we reduce the problem to sampling from the so-called Gibbs distribution.
For the sampling task, we use block dynamics, where the blocks are given by a so-called star colouring \cite{ACKKR04} using a constant number of colours.
A star colouring is a proper colouring and any two colour classes form a star forest.
We treat each colour class as a block, and show that the block dynamics is rapidly mixing when $\lambda$ is sufficiently small.
In fact, via standard comparison arguments, we can show that the single-site Glauber dynamics is rapidly mixing as well.

\subsection{Markov chain preliminaries}

Let $P$ be the transition matrix of a discrete-time Markov chain on a finite state space $\Omega$.
Thus, $P(x,y)\ge 0$ and $\sum_{y\in\Omega}P(x,y)=1$ for every $x\in\Omega$.
We write $P^t(x,\cdot)$ for the distribution of the chain after $t$ steps when it is started from $x$.
A probability distribution $\pi$ on $\Omega$ is \emph{stationary} for $P$ if $\pi P=\pi$.

The chain $P$ is \emph{irreducible} if, for every $x,y\in\Omega$, there is some $t\ge 0$ such that $P^t(x,y)>0$.
It is \emph{aperiodic} if $\gcd\set{t\ge 1:P^t(x,x)>0}=1$ for every $x\in\Omega$.
Every finite irreducible chain has a unique stationary distribution $\pi$.
It converges to $\pi$ from every initial state if it is also aperiodic.
The chain is \emph{reversible} with respect to $\pi$ if $\pi(x)P(x,y)=\pi(y)P(y,x)$ for all $x,y\in\Omega$.
Let $L^2(\pi)$ denote the space of real-valued functions on $\Omega$ equipped with the inner product $\inner{\pi}{f,g}=\sum_{x\in\Omega}\pi(x)f(x)g(x)$.
For $f\in L^2(\pi)$, its norm is $\norm{\pi,2}{f}:=\sqrt{\inner{\pi}{f,f}}$.
Reversibility implies that $P$ is self-adjoint on $L^2(\pi)$, and hence that all its eigenvalues are real.

For probability distributions $\mu$ and $\nu$ on $\Omega$, their total-variation distance is $\DTV{\mu}{\nu}:=\frac12\sum_{x\in\Omega}\abs{\mu(x)-\nu(x)}$.
For an irreducible and aperiodic chain, define its mixing time by $\mixingtime{\eps}:=\min\set{t\ge 0:\max_{x\in\Omega}\DTV{P^t(x,\cdot)}{\pi}\le\eps}$ for $0<\eps<1$.

In this paper we consider Glauber dynamics, which updates a uniformly at random variable at each step, conditioned on all other variables being fixed,
and more generally, block dynamics, which updates a random block of variables instead.
It is known that for Glauber or block dynamics, there is no negative eigenvalue \cite{DGU14}.
Moreover, they are reversible and, for the problems considered in this paper, irreducible.
Let $\lambda_2(P)$ denote the second largest eigenvalue of $P$, and $\operatorname{gap}(P)\defeq1-\lambda_2(P)$.
Let $\pimin\defeq\min_{x\in\Omega}\pi(x)$.
Then it is standard \cite{LP17} that, for $0<\eps<1$,
\begin{align}\label{eq:spectral-mixing-bound}
    \mixingtime{\eps} = O\left( \frac{1}{\operatorname{gap}(P)}    \log\frac{1}{\eps\pimin}\right).
\end{align}

For the hard-core model, its Gibbs distribution is defined by $\mu(I)\propto \lambda^{\abs{I}}$ for every independent set $I\in\+I(G)$.
If we can generate a sample from $\mu$ in polynomial time, then we can use the standard self-reduction \cite{JVV86} or simulated annealing \cite{SVV09} to approximate $\ZHC$ in polynomial time as well.
Thus, in light of \eqref{eq:spectral-mixing-bound}, the main task is to bound the spectral gap of Glauber or block dynamics.

\subsection{A sufficient condition to bound the spectral gap}

Now we derive a new sufficient condition to bound the spectral gap of Glauber dynamics.
A key quantity is the Hirschfeld–Gebelein–Rényi maximal correlation \cite{Ren59}, which is a measure of correlation between two random variables.
For two real-valued random variables $U$ and $V$ with positive variance, define 
\begin{align*}
  \mathrm{Cov}(U,V)&\defeq\Ex[(U-\Ex[U])(V-\Ex[V])],  \\
  \mathrm{Corr}(U,V)&\defeq\frac{\mathrm{Cov}(U,V)}{\sqrt{\operatorname{Var}(U)\operatorname{Var}(V)}}.
\end{align*}

\begin{definition}[Hirschfeld–Gebelein–Rényi maximal correlation]

Let $X$ and $Y$ be jointly distributed random variables taking values in finite sets $\+X$ and $\+Y$, respectively.
Their Hirschfeld–Gebelein–R\'enyi maximal correlation is
\begin{align*}
\rho_{\mathrm{HGR}}(X,Y) := \max_{\substack{f: \+X \to \mathbb{R},\,\operatorname{Var}(f(X))>0\\ g: \+Y \to \mathbb{R},\,\operatorname{Var}(g(Y))>0}} \mathrm{Corr}(f(X),g(Y)).
\end{align*}
\end{definition}

\begin{remark}
    Let $H_{X} = \{f:\+X \to \mathbb{R}\}$ and $H_{Y} = \{g:\+Y \to \mathbb{R}\}$. Let $(u_i)_{1 \le i \le s}$ and $(v_j)_{1 \le j \le t}$ be orthonormal bases of $H_{X}$ and $H_{Y}$ respectively in $L^2(\mu)$, where $u_1$ and $v_1$ are the all-$1$ functions. The HGR maximal correlation $\rho_{\mathrm{HGR}}$ is exactly the second largest singular value of an $s$-by-$t$ matrix $R=(r_{ij})$, where $r_{i,j} = \inner{\mu}{u_i,v_j}$. Therefore, for random variables $(X,Y) \sim \bigotimes_{i=1}^k \mu_i$, the HGR maximal correlation between $X=(X_1,X_2,\ldots,X_k)$ and $Y=(Y_1,Y_2,\ldots,Y_k)$ satisfies
    \begin{align}\label{eq:max-principle}
        \rho_{\mathrm{HGR}}(X,Y) = \max_{1 \le i\le k} \rho_{\mathrm{HGR}}(X_i,Y_i).
    \end{align}
\end{remark}

For any $f\in L^2(\mu)$, $(X_1,X_2)\sim\mu$, and $i\in\{1,2\}$, let $P_i$ be the averaging operator over the $i$-th coordinate defined by
$P_i f\defeq\Ex_\mu[f\mid X_{-i}]$,
where $X_{-i}$ denotes the other coordinate.
Thus, $P_i$ is the orthogonal projection in $L^2(\mu)$ onto the functions of $X_{-i}$.
Let $L_i\defeq I-P_i$.
\begin{lemma}\label{lem:soft-Cauchy}
    Let $\mu$ be a distribution over a finite space $\+X_1 \times \+X_2$. For any $f:\+X_1 \times \+X_2 \to \mathbb{R}$, it holds
    \begin{align*}
    \inner{\mu}{L_1 f, L_2f} \ge -\rho \norm{\mu, 2}{L_1 f} \norm{\mu, 2}{L_2 f},
    \end{align*}
    where $(X,Y) \sim \mu$ and $\rho = \rho_{\mathrm{HGR}}(X,Y)$.
\end{lemma}

\begin{proof}
    Without loss of generality, we assume $\Ex_\mu[f] = 0$.
    If $f=0$, then the claim is immediate.
    Let $u = \frac{\norm{\mu,2}{P_1 f}^2}{\norm{\mu,2}{f}^2}$, $v = \frac{\norm{\mu,2}{P_2 f}^2}{\norm{\mu,2}{f}^2}$.
    If $P_1f=0$ or $P_2f=0$, then
    \begin{align*}
        \inner{\mu}{L_1f,L_2f}
        =\tp{1-u-v}\norm{\mu,2}{f}^2\geq0,
    \end{align*}
    and the claim follows.
    We may therefore assume that both $P_1f$ and $P_2f$ are nonzero.
    Let $\alpha \defeq \frac{\inner{\mu}{P_1f, P_2 f}}{\norm{\mu,2}{P_1 f} \norm{\mu,2}{P_2 f}} = \frac{\inner{\mu}{P_1 f,P_2 f}}{\sqrt{uv} \norm{\mu,2}{f}^2}$. 
    Since $P_1f$ and $P_2f$ are centered functions of $X_2$ and $X_1$, respectively, the definition of HGR maximal correlation gives $\abs{\alpha}\leq\rho$.
    By the definitions and $\inner{\mu}{f, P_i f} = \norm{\mu,2}{P_i f}^2$, the inner product $\inner{\mu}{L_1f, L_2f}$ satisfies
    \begin{align*}
        \inner{\mu}{L_1 f, L_2 f} = \tp{1-u-v+\alpha \sqrt{uv}} \norm{\mu,2}{f}^2.
    \end{align*}
    Note that $\norm{\mu,2}{L_1 f}^2 = (1-u)  \norm{\mu,2}{f}^2$ and $\norm{\mu,2}{L_2 f}^2 = (1-v)  \norm{\mu,2}{f}^2$. Therefore, it suffices to show
    \begin{align}\label{eq:target-soft-Cauchy}
        1-u-v+\alpha \sqrt{uv} + \rho \sqrt{(1-u)(1-v)} \ge 0.
    \end{align}
    Consider the following $3$-by-$3$ Gram matrix of the vectors $\frac{f}{\norm{\mu,2}{f}}$, $\frac{P_1 f}{\norm{\mu,2}{P_1 f}}$ and $\frac{P_2 f}{\norm{\mu,2}{P_2 f}}$:
    \begin{align*}
        G = 
        \begin{bmatrix}
        1 & \sqrt{u} & \sqrt{v}\\
        \sqrt{u} & 1 & \alpha\\
        \sqrt{v} & \alpha & 1
        \end{bmatrix}.
    \end{align*}
    Since Gram matrices are positive semidefinite, the determinant $\det(G) = 1-u-v + 2 \sqrt{uv} \alpha - \alpha^2 \ge 0$, which implies $\alpha \ge \sqrt{uv} - \sqrt{(1-u)(1-v)}$. Together with $\alpha \le \rho$, the left-hand side of~\eqref{eq:target-soft-Cauchy} becomes
    \begin{align*}
    \text{LHS of~\eqref{eq:target-soft-Cauchy}} & \ge 1 - u - v + \tp{\sqrt{uv} - \sqrt{(1-u)(1-v)}} \tp{\sqrt{uv} + \sqrt{(1-u)(1-v)}} = 0. \qedhere
    \end{align*}
\end{proof}

For a real matrix $R$, its matrix $2$-norm is $\norm{2}{R}:=\max_{x\ne 0}\norm{2}{Rx}/\norm{2}{x}$, where $\norm{2}{x}:=(\sum_i x_i^2)^{1/2}$ is the Euclidean norm.
A pinning $\tau$ on a subset $S \subseteq [n]$ is an assignment of values to the variables in $S$,
and feasible pinnings are those that have non-zero probability under $\mu$.
The main result of this subsection is the next lemma.

\begin{lemma}\label{lem:dobrushin}
Let $\mu$ be a distribution over a finite set $\+X_1 \times \+X_2 \times \ldots \times \+X_n$. For each $i \neq j$, let $r_{ij}$ be the maximum value of $\rho_{\mathrm{HGR}}(X^\tau_i,X^\tau_j)$ over all feasible pinnings of $\tau$ on $[n] \setminus \{i,j\}$ where $(X^\tau_i,X^\tau_j) \sim \mu^\tau_{\{i,j\}}$. Moreover, let $R \defeq (r_{i,j})_{1 \le i,j \le n}$, where $r_{i,i}=0$. 

If $\norm{2}{R} < 1$ and the Glauber dynamics $P$ on $\mu$ is irreducible, it satisfies
\begin{align*}
    \lambda_2(P) \le 1 - \frac{1-\norm{2}{R}}{n}.
\end{align*}
\end{lemma}

\begin{proof}
    Let $f$ be the eigenvector corresponding to the second largest eigenvalue $\lambda_2(P)$.
    Again, let $L_i \defeq I - P_i$ for $i\in[n]$, where $P_i$ is the averaging operator on the $i$-th coordinate. Then, as $P=\frac{1}{n} \sum_{i=1}^n P_i$, we have $\sum_{i=1}^n L_i = n(I-P)$.
    For fixed $i\neq j$, let $Y=X_{[n]\setminus\{i,j\}}$ and condition on a feasible assignment $Y=\tau$.
    Applying~\Cref{lem:soft-Cauchy} to the resulting two-coordinate distribution gives
    \begin{align*}
        \Ex_\mu[L_ifL_jf\mid Y=\tau]
        \geq-r_{ij}
        \sqrt{\Ex_\mu[(L_if)^2\mid Y=\tau]}
        \sqrt{\Ex_\mu[(L_jf)^2\mid Y=\tau]}.
    \end{align*}
    Averaging over $Y$ and applying Cauchy--Schwarz yields
    \begin{align}\label{eq:conditional-soft-Cauchy}
        \inner{\mu}{L_if,L_jf}
        &\geq-r_{ij}\Ex_T\sqtp{
        \sqrt{\Ex_\mu[(L_if)^2\mid Y=\tau]}
        \sqrt{\Ex_\mu[(L_jf)^2\mid Y=\tau]}}\nonumber\\
        &\geq-r_{ij}\norm{\mu,2}{L_i f}\norm{\mu,2}{L_j f}.
    \end{align}
    Since $P$ is irreducible, $\lambda_2(P) < 1$.
    Summing~\eqref{eq:conditional-soft-Cauchy} over $i\neq j$, we obtain
    \begin{align}\label{eq:quadratic}
        \nonumber \norm{\mu, 2}{\sum_{i=1}^n L_i f}^2 = \sum_{i=1}^n \norm{\mu,2}{L_i f}^2 + \sum_{i \neq j} \inner{\mu}{L_i f,L_j f} &\ge \sum_{i=1}^n \norm{\mu,2}{L_i f}^2 - \sum_{i \neq j} r_{ij} \norm{\mu,2}{L_i f} \norm{\mu,2}{L_j f}\\
        &\ge (1-\norm{2}{R}) \sum_{i=1}^n \norm{\mu,2}{L_i f}^2,
    \end{align}
    where the last inequality is due to the Cauchy–Schwarz inequality and the definition of matrix $2$-norm.

    Since $\sum_{i=1}^n L_i = n(I-P)$, the left-hand side of~\eqref{eq:quadratic} is exactly $n^2 (1-\lambda_2(P))^2 \norm{\mu,2}{f}^2$. Similarly, the right-hand side of~\eqref{eq:quadratic} satisfies
    \begin{align*}
        (1-\norm{2}{R}) \sum_{i=1}^n \norm{\mu,2}{L_i f}^2 = (1-\norm{2}{R}) \sum_{i=1}^n \inner{\mu}{f,L_i f} = (1-\norm{2}{R}) n (1-\lambda_2(P)) \norm{\mu,2}{f}^2. 
    \end{align*}
    Combining with~\eqref{eq:quadratic}, this completes the proof.
\end{proof}

To apply \Cref{lem:dobrushin}, we need to bound the HGR maximal correlation between every pair of coordinates under arbitrary feasible pinnings of the remaining coordinates.
By the tensorisation property in~\eqref{eq:max-principle}, this is particularly convenient for block dynamics when the conditional joint distribution of two blocks factorises, as we shall see next.
These HGR correlations can also be interpreted as the nontrivial singular values of the weighted codimension-two links in the simplicial complex of feasible partial configurations.
This interpretation is related to the trickle-down and local-to-global approaches \cite{Opp18,AL20}, although \Cref{lem:dobrushin} does not seem to follow directly from them.
In principle, the matrix trickle-down theorem \cite{ALO21} should be able to recover it, but it seems non-trivial and would involve constructing particular matrices, which we will not do here.

\begin{remark}\label{rem:LO25}
  After the release of the manuscript, Zongchen Chen pointed out to us that the matrix $R$ defined above is entrywise the same as the pairwise  influence matrix introduced by Leake and Oveis Gharan~\cite{LO25}.
  See \cite[Theorem 33.12]{PW25}.
  Thus, our refined spectral gap result, \Cref{lem:dobrushin}, addresses their open question of improving the mixing time from $O\tp{n^{1+1/\delta}}$ to $O_\delta\tp{n^2}$.
\end{remark}

\subsection{The algorithm}
In this section we prove \Cref{thm:HC}.
A key ingredient is the star colouring, defined next.
We allow isolated vertices as stars.

\begin{definition}[star colouring]\label{def:star-colouring}
Let $G=(V,E)$ be an undirected graph and $\sigma \in [q]^V$ be a proper colouring. The colouring $\sigma$ is a \emph{star colouring} if for any two colours $i$ and $j$, each connected component in the subgraph induced by vertices with either colour $i$ or colour $j$ is a star.
\end{definition}

Star colourings can be found in polynomial time for planar graphs for $20$ colours.

\begin{proposition}[\cite{ACKKR04}]\label{prop:star-colouring}
    For any planar graph, there is a star colouring using $20$ colours, and it can be found in quadratic time.\footnote{In \cite{ACKKR04}, no explicit time complexity is given. The main bottleneck of the algorithm is to find an acyclic $5$-colouring and a $4$-colouring of a planar graph. Both tasks can be done in quadratic time \cite{Bor79,RSST97}.}
\end{proposition}

Now we are ready to prove \Cref{thm:HC}.
The algorithm first finds a star colouring of the input planar graph, and then runs the block dynamics on the colour classes.

\begin{proof}[Proof of~\Cref{thm:HC}]
    By~\Cref{prop:star-colouring}, we can find a star colouring with at most $20$ colours in polynomial time.
    Equivalently, the vertex set $V$ is divided into $k=20$ blocks $C_1,C_2,\ldots,C_k$.
    We consider the uniform block dynamics on $G$ with these blocks.
    It is easy to see that the block dynamics is irreducible and can be implemented in polynomial time.
    Moreover, $\log\pimin^{-1}=O(n)$, where $n$ is the number of vertices in $G$.
    By treating the partial configurations on each block as a spin state and applying~\Cref{lem:dobrushin}, it remains to show that the HGR maximal correlation between any two blocks $C_i$ and $C_j$, under arbitrary feasible pinnings outside these blocks, is at most $\frac{1}{2(k-1)}$.
    Pinning a vertex occupied is the same as removing the vertex and its neighbours, and pinning a vertex unoccupied is the same as removing the vertex.
    Thus, we can remove the pinnings and get a smaller graph, which is still a star forest.
    
    By the definition of star colourings, the graph remaining after applying the pinnings is a disjoint union of stars.
    By~\eqref{eq:max-principle}, it suffices to consider a single star $K_{1,d}$ with $d\ge 1$.
    Let $u$ be its center and $v_1,v_2,\ldots,v_d$ its leaves.
    The center and the leaves have different colours because a star colouring is proper.
    It suffices to show that for every nonconstant $f : \{0,1\} \to \mathbb{R}$ and $g : \{0,1\}^d \to \mathbb{R}$,
    \begin{align}\label{eq:cov-bound-target}
        \mathrm{Cov}(f(X),g(Y)) \le \frac{1}{2(k-1)} \sqrt{\Var{}{f(X)} \Var{}{g(Y)}},
    \end{align}
    where $(X,Y)$ are drawn from the Gibbs distribution of the hard-core model on this star.
    Indeed, \eqref{eq:cov-bound-target} implies $\mathrm{Corr}(f(X),g(Y))\le 1/(2(k-1))$.
    We want to apply \Cref{lem:dobrushin}.
    Taking the maximum over $f$ and $g$ gives $r_{ij}\le 1/(2(k-1))$ for all $i\ne j$.
    Since $R$ is symmetric, nonnegative, and has zero diagonal, its matrix norm is at most its maximum row sum, and hence $\norm{2}{R}\le (k-1)/(2(k-1))=1/2<1$.
    Thus,~\eqref{eq:cov-bound-target} is sufficient for applying~\Cref{lem:dobrushin}.

    Let $p:=\Pr(X=1)$ be the marginal probability that the center is occupied and $q:=\Pr(Y\ne\*0_d)$ be the marginal probability that the leaves are not all empty.
    Since $X=1$ and $Y\ne\*0_d$ are mutually exclusive,
    \begin{align*}
        \Pr(X=1,Y=\*0_d)=p,\qquad
        \Pr(X=0,Y\ne\*0_d)=q,\qquad
        \Pr(X=0,Y=\*0_d)=1-p-q.
    \end{align*}
    We may therefore normalize $f$ by $\Ex[f(X)]=0$ and $f(1)-f(0)=1$, which gives $f(0)=-p$ and $f(1)=1-p$.

    Let $h(Y):=\Ex[f(X)\mid Y]$.
    For $g$ with $\Ex g=0$, 
    the tower property gives
    \begin{align*}
        \mathrm{Cov}(f(X),g(Y))=\Ex[h(Y)g(Y)].
    \end{align*}
    Moreover,
    \begin{align*}
        h(y)=
        \begin{cases}
            \displaystyle \frac{p}{1-q}(1-p)+\frac{1-p-q}{1-q}(-p)
              =\frac{pq}{1-q}, & y=\*0_d,\\[6pt]
            -p, & y\ne\*0_d.
        \end{cases}
    \end{align*}
    By Cauchy--Schwarz, $\Ex[h(Y)g(Y)]\le\sqrt{\Var{}{h(Y)}\Var{}{g(Y)}}$, with equality when $g$ is proportional to $h$.
    A maximizing choice that satisfies $\Ex[g(Y)]=0$ is $g(\*0_d)=q$ and $g(y)=-(1-q)$ for every $y\ne\*0_d$.
    A direct calculation gives
    \begin{align*}
        \mathrm{Cov}(f(X),g(Y))
        &=p(1-p)q+q(-p)(-(1-q))+(1-p-q)(-p)q=pq,\\
        \Var{}{f(X)}&=p(1-p),\\
        \Var{}{g(Y)}&=q(1-q).
    \end{align*}
    Consequently, the maximum correlation is
    \begin{align*}
        \rho_{\mathrm{HGR}}(X,Y)
        =\frac{pq}{\sqrt{p(1-p)q(1-q)}}
        =\sqrt{\frac{pq}{(1-p)(1-q)}}.
    \end{align*}
    The partition function of the star is $Z=\lambda+(1+\lambda)^d$, so $p=\lambda/Z$ and $q=((1+\lambda)^d-1)/Z$.
    It follows that
    \begin{align*}
        \rho_{\mathrm{HGR}}(X,Y)
        =\sqrt{\frac{\lambda((1+\lambda)^d-1)}{(1+\lambda)^{d+1}}}
        \le \sqrt{\frac{\lambda}{1+\lambda}}.
    \end{align*}
    Taking $\lambda_0=1/(4(k-1)^2-1)$ makes the last expression at most $1/(2(k-1))$ whenever $\lambda\le\lambda_0$.
    For $k=20$, this gives $\lambda_0=1/1443$ and completes the proof.
\end{proof}

\subsection{Comparison with Glauber dynamics}  \label{sec:Glauber}

The spectral gap analysis in the last subsection also implies a spectral gap lower bound for the single-site Glauber dynamics via a standard comparison argument.

\begin{theorem}\label{thm:glauber-gap}
    Let $\lambda_0=1/1443$ and let $G$ be a planar graph with $n$ vertices.
    For $\lambda\le \lambda_0$, let $P$ be the transition matrix of the Glauber dynamics for the hard-core distribution with fugacity $\lambda$.
    Then,
    \begin{align*}
        \operatorname{gap}(P)\ge \frac{1}{2n}.
    \end{align*}
\end{theorem}

\begin{proof}
    By \Cref{prop:star-colouring}, let $C_1,\ldots,C_k$ be the colour classes of a star colouring of $G$, where $k=20$.
Let $P_{C_j}$ be the averaging operator on the block $C_j$, and let
$P_{\mathrm{block}}=\frac{1}{k}\sum_{j=1}^k P_{C_j}$
be the transition operator of the corresponding block dynamics.
Both $P$ and $P_{\mathrm{block}}$ are reversible with respect to the hard-core distribution $\mu$.
For any reversible transition operator $Q$ with stationary distribution $\mu$, the variational characterisation of the spectral gap is
\begin{align}\label{eq:gap-variational-characterisation}
    \operatorname{gap}(Q)
    =\inf_{f:\,\Var{\mu}{f}>0}
    \frac{\inner{\mu}{f,(I-Q)f}}{\Var{\mu}{f}}.
\end{align}
See \cite{LP17}.
As shown in the proof of~\Cref{thm:HC}, an application of~\Cref{lem:dobrushin} gives
\begin{align}\label{eq:block-gap-for-comparison}
    \operatorname{gap}(P_{\mathrm{block}})\geq\frac{1}{2k}.
\end{align}

For $v\in V$, let $P_v$ be the averaging operator on the spin at $v$.
Then $P=\frac{1}{n}\sum_{v\in V}P_v$.
Since each $C_j$ is an independent set, conditioned on the spins outside $C_j$, the spins in $C_j$ are independent.
The tensorisation of variance for this conditional product distribution gives, for every $f\in L^2(\mu)$,
\begin{align*}
    \inner{\mu}{f,(I-P_{C_j})f}
    \leq\sum_{v\in C_j}\inner{\mu}{f,(I-P_v)f}.
\end{align*}
Indeed, the left-hand side is the expected conditional variance of $f$ given the spins outside $C_j$, and each summand on the right-hand side is the expected conditional variance obtained by resampling the spin at $v$ while fixing all other spins.
Summing over the colour classes yields
\begin{align*}
    \inner{\mu}{f,(I-P_{\mathrm{block}})f}
    &\leq\frac{1}{k}\sum_{v\in V}\inner{\mu}{f,(I-P_v)f}=\frac{n}{k}\inner{\mu}{f,(I-P)f}.
\end{align*}
Therefore, for every nonconstant $f$,
\begin{align*}
    \frac{\inner{\mu}{f,(I-P)f}}{\Var{\mu}{f}}
    &\geq\frac{k}{n}\cdot
    \frac{\inner{\mu}{f,(I-P_{\mathrm{block}})f}}{\Var{\mu}{f}}\geq\frac{k}{n}\operatorname{gap}(P_{\mathrm{block}})\geq\frac{1}{2n},
\end{align*}
where the last inequality follows from~\eqref{eq:block-gap-for-comparison}.
Taking the infimum over $f$ with positive variance and applying~\eqref{eq:gap-variational-characterisation} proves the claim.
\end{proof}

\subsection{Generalisation to anti-ferromagnetic 2-spin systems}  

Now we can generalise \Cref{thm:HC} to anti-ferromagnetic 2-spin systems, namely, we prove Item \eqref{item:2-spin-easy} of \Cref{thm:2-spin}.

\begin{proof}[Proof of Item~\eqref{item:2-spin-easy} of~\Cref{thm:2-spin}]
    If $\beta\gamma>1$, the claim follows from the FPRAS for ferromagnetic two-spin systems in~\cite{GJP03}.
    If $\beta\gamma=1$, the partition function factorises and can be computed exactly.
    It remains to consider $\beta\gamma<1$.
    Since we are outside the case in Item~\eqref{item:2-spin-hard}, we have $\beta\ge 1$, and consequently $\gamma<1$.

    Let $C_1,\ldots,C_k$ be the colour classes of a star colouring given by~\Cref{prop:star-colouring}, where $k\le 20$.
    We run the uniform block dynamics whose blocks are these colour classes.
    Fix two blocks $C_i$ and $C_j$ and an arbitrary feasible pinning of all other blocks.
    For an unpinned vertex $v$, let $a_v$ and $b_v$ be the numbers of its pinned neighbours assigned spins $0$ and $1$, respectively.
    After constant factors are removed, the pinning replaces the activity at $v$ by
    \begin{align*}
        \lambda_v=\lambda\beta^{-a_v}\gamma^{b_v}\le\lambda.
    \end{align*}
    When $\gamma=0$ and $b_v>0$, the vertex $v$ is forced to spin $0$ and may be deleted.
    Thus, the conditional distribution on $C_i\cup C_j$ is a two-spin system on a disjoint union of stars, with possibly nonuniform activities satisfying $0<\lambda_v\le\lambda$.

    Consider one such star with center spin $X$ and leaf-spin vector $Y$.
    Let $Q(Y):=\Pr(X=1\mid Y)$.
    For a leaf configuration $y$, the ratio of the two conditional weights of the center is
    \begin{align*}
        \frac{Q(y)}{1-Q(y)}
        =\lambda_u\prod_{v:y_v=0}\frac{1}{\beta}
        \prod_{v:y_v=1}\gamma
        \le \lambda,
    \end{align*}
    where $u$ is the center.
    Hence, with $a:=\lambda/(1+\lambda)$, we have $0\le Q(Y)\le a$.

    Let $p:=\Pr(X=1)=\Ex[Q(Y)]$.
    Since $X$ is binary, every nonconstant centered function of $X$ is a scalar multiple of $X-p$.
    We may therefore fix the function of $X$ to be $X-p$ when computing the HGR maximal correlation.
    For every function $g$ satisfying $\Ex[g(Y)]=0$, the tower property gives
    \begin{align*}
        \Ex[(X-p)g(Y)]
        &=\Ex\!\left[\Ex[X-p\mid Y]g(Y)\right]\\
        &=\Ex[(Q(Y)-p)g(Y)].
    \end{align*}
    Hence, Cauchy--Schwarz implies
    \begin{align*}
        \mathrm{Corr}(X,g(Y))^2
        &=\frac{\Ex[(Q(Y)-p)g(Y)]^2}{p(1-p)\Var{}{g(Y)}}\le \frac{\Var{}{Q(Y)}}{p(1-p)}.
    \end{align*}
    If $Q(Y)$ is nonconstant, equality is attained by taking $g(Y)=Q(Y)-p$.
    If $Q(Y)$ is constant, both sides are zero.
    Taking the maximum over $g$ therefore gives
    \begin{align*}
        \rho_{\mathrm{HGR}}(X,Y)^2
         =\frac{\Var{}{Q(Y)}}{\Var{}{X}}     =\frac{\Var{}{Q(Y)}}{p(1-p)}.
    \end{align*}
    The bound $Q(Y)\le a$ implies
    \begin{align*}
        \Var{}{Q(Y)}
        \le a\Ex[Q(Y)]-p^2
        =p(a-p)
        \le ap(1-p).
    \end{align*}
    Consequently,
    \begin{align}\label{eq:two-spin-star-correlation}
        \rho_{\mathrm{HGR}}(X,Y)\le \sqrt{\frac{\lambda}{1+\lambda}}.
    \end{align}

    From here, the proof is the same as in \Cref{thm:HC}.
    We can choose $\lambda_0=\frac{1}{4(20-1)^2-1}=\frac{1}{1443}$, 
    and apply \Cref{lem:dobrushin} to show that the spectral gap of the block dynamics is at least $1/(2k)$.

    Each block update can be sampled in linear time because every colour class is an independent set and its spins are conditionally independent.
    Moreover, for fixed $\beta$, $\gamma$, and $\lambda$, $\log\pimin^{-1}=O(n)$ on planar graphs.
    Thus,~\eqref{eq:spectral-mixing-bound} gives a polynomial-time approximate sampler.
    The same argument applies after any feasible vertex pinning.
    The standard sampling-to-counting reduction~\cite{JVV86} then gives an FPRAS for the partition function and completes the proof.
\end{proof}

\section{NP-hardness of planar colouring}

In this section we prove \Cref{thm:colouring}.
The key is to find the right \NP{}-hard problem to reduce from.
Our starting point is a problem, called \emph{graph coalition structure generation} (GCSG), that is shown to be \NP{}-hard by Voice, Polukarov, and Jennings~\cite{VPJ12} in planar graphs (see also \cite[Theorem~5 and Appendix~I]{BKKZ13} for an alternative and shorter proof).  
A \emph{coalition structure} is a partition of the vertex set into disjoint coalitions,
but the number of coalitions is not fixed.
We consider the following variant, where edges are assigned signed but unit weights.

\prob{GCSG}
{A signed planar graph $G=(V,E,w)$, where $w\colon E\to\{-1,+1\}$.}
{The maximum score over all coalition structures $\Pi$, where
\begin{align*}
    \mathrm{score}_{G}(\Pi)
    =
    \sum_{C\in\Pi}\;
    \sum_{\substack{uv\in E,\; u,v\in C}}
    w(uv).
\end{align*}}

\begin{proposition}
\label{prop:hardness-WGG}
    GCSG is \NP{}-hard on planar graphs.
\end{proposition}

\begin{proof}
Voice, Polukarov, and Jennings~\cite{VPJ12} showed that GCSG is \NP{}-hard on planar graphs with integer edge weights of polynomial magnitude.
We reduce their hard instances to the signed unit-weight version.
Zero-weight edges may first be deleted.

Let $G=(V,E,w)$ be one of their instances.
For every edge $uv\in E$, replace $uv$ by $\abs{w(uv)}$ internally vertex-disjoint paths $uxv$ of length two.
If $w(uv)>0$, give both edges of every such path weight $+1$.
If $w(uv)<0$, give $ux$ weight $+1$ and $xv$ weight $-1$.
Denote the resulting signed unit-weight graph by $\widetilde G$.

Fix a partition $\Pi$ of the original vertex set $V$.
For a path with weights $(+1,+1)$, the maximum contribution over the choice of the coalition containing $x$ is $2$ when $u$ and $v$ are together and $1$ when they are separate.
Thus, this contribution is $1+\mathbf 1\{u,v\text{ are together}\}$.
For a path with weights $(+1,-1)$, the corresponding maximum contributions are $0$ and $1$, respectively.
Thus, this contribution is $1-\mathbf 1\{u,v\text{ are together}\}$.
The internal vertices of different paths are distinct, so their coalitions can be chosen independently.
Consequently,
\begin{align}\label{eq:unit-weight-simulation}
    \max_{\substack{\widetilde\Pi:\;\widetilde\Pi|_V=\Pi}}
    \mathrm{score}_{\widetilde G}(\widetilde\Pi)
    =\mathrm{score}_{G}(\Pi)+\sum_{e\in E}\abs{w(e)}.
\end{align}
Taking the maximum over $\Pi$ shows that the two optimum values differ by the known additive constant $\sum_{e\in E}\abs{w(e)}$.

The replacement preserves planarity because all paths replacing an edge can be drawn inside a thin neighbourhood of that edge.
It has polynomial size because the weights in the hard instances of~\cite{VPJ12} have polynomial magnitude.
This proves the proposition.
\end{proof}

Let $q\ge 4$ be fixed.
Consider the following restriction of GCSG where the number of coalitions is bounded by $q$.

\prob{\textsc{$q$-cut}}
{A weighted planar graph $G=(V,E,w)$, where $w\colon E\to\{-1,+1\}$.}
{The maximum score over all maps $\sigma\colon V\to[q]$, where
\begin{align*}
    \mathrm{score}_{G}(\sigma)
    =
    \sum_{\substack{uv\in E,\;\sigma(u)=\sigma(v)}} w(uv).
\end{align*}}

\begin{lemma}\label{lem:q-cut-hardness}
    For every fixed integer $q\geq 4$, computing \textsc{$q$-cut} is \NP{}-hard.
\end{lemma}

\begin{proof}
Let $\Pi=(C_1,C_2,\ldots,C_k)$ be an optimal coalition structure for an instance of GCSG.
We may assume that every subgraph $G[C_i]$ is connected.
Indeed, replacing a disconnected coalition by the vertex sets of its connected components does not change the score.

Contract every $C_i$ to one vertex and discard loops.
The resulting graph is planar, so its underlying simple graph has a proper four-colouring~\cite{RSST97}.
Give all vertices in $C_i$ the colour assigned to its contracted vertex.
An edge is monochromatic in this four-colouring exactly when its endpoints belong to the same coalition of $\Pi$.
The four-colouring therefore has the same score as $\Pi$.
It is also a map to $[q]$ because $q\ge4$.

Conversely, every map $V\to[q]$ defines a coalition structure through its colour classes.
The optimum values of GCSG and \textsc{$q$-cut} are therefore equal.
The claim follows from~\Cref{prop:hardness-WGG}.
\end{proof}

\subsection{Reduction to counting colourings}

To reduce from \textsc{$q$-cut} to counting colourings, we just need to amplify the number of colourings coming from the optimal coalition structure, so that the optimal contributes a significant fraction of all proper colourings.

\begin{proof}[Proof of~\Cref{thm:colouring}]

For any signed unit-weight instance $G=(V,E,w)$ of \textsc{$q$-cut},
let $n=\abs{V}$, and let $m_+$ and $m_-$ be the numbers of positive and negative edges, respectively.
For positive integers $s$ and $t$ to be chosen below, construct an unweighted graph $H=(V',E')$ as follows.
Replace every positive edge by $s$ internally vertex-disjoint paths of length two.
Replace every negative edge by $t$ internally vertex-disjoint paths of length three.
All internal vertices introduced for different paths are distinct.
The construction preserves planarity.

For any colouring $\sigma\in[q]^{V}$, the number of proper colourings $\tau\in[q]^{V'}$ satisfying
$\tau_{V}=\sigma$ is $w(\sigma):=\prod_{x\in\{+,-\},\,y\in\{\mathrm{same},\mathrm{diff}\}}c_{x,y}^{m_{x,y}}$, where $m_{x,y}$ denotes the number of edges of sign $x$
whose endpoints receive the same or different colours in $\sigma$, according
to whether $y=\mathrm{same}$ or $y=\mathrm{diff}$. The coefficients
$c_{x,y}$ are given by
\begin{align*}
    c_{+,\mathrm{same}}
    &\defeq(q-1)^s,
    &
    c_{+,\mathrm{diff}}
    &\defeq(q-2)^s,
    \\
    c_{-,\mathrm{same}}
    &\defeq\tp{(q-1)(q-2)}^t,
    &
    c_{-,\mathrm{diff}}
    &\defeq\tp{(q-1)+(q-2)^2}^t.
\end{align*}

Choose integers $s$ and $t$ such that
\begin{enumerate}
    \item $s=\Theta(n)$ and $t=\Theta(n)$;
    \item $\tp{\frac{q-1}{q-2}}^s\geq 100q^{100n}$;
    \item $2\tp{\frac{q-1}{q-2}}^s\geq\tp{\frac{(q-1)+(q-2)^2}{(q-1)(q-2)}}^t \geq \tp{\frac{q-1}{q-2}}^s$.
\end{enumerate}
The existence of $s$ and $t$ is due to $1<\frac{(q-1)+(q-2)^2}{(q-1)(q-2)}<\frac{q-1}{q-2}<2$.

Let $K:=\max_{\sigma\in[q]^{V}}
\mathrm{score}_{G}(\sigma)$, and define
$Z_0:=(q-2)^{sm_+}\tp{(q-1)(q-2)}^{tm_-}$.
Set
\begin{align*}
    A:=\frac{q-1}{q-2},
    \qquad
    B:=\frac{(q-1)+(q-2)^2}{(q-1)(q-2)},
    \qquad
    \theta:=\frac{B^t}{A^s}.
\end{align*}
Our choice of $s$ and $t$ ensures that $1\leq\theta\leq2$.
For $\sigma\in[q]^V$, let $p(\sigma)$ and $r(\sigma)$ be the numbers of monochromatic positive and negative edges, respectively.
Then $\mathrm{score}_G(\sigma)=p(\sigma)-r(\sigma)$, and the number of extensions of $\sigma$ to a proper colouring of $H$ is
\begin{align*}
    w(\sigma)
    &=Z_0 A^{s\cdot p(\sigma)}B^{t(m_--r(\sigma))}\\
    &=Z_0 A^{s(\mathrm{score}_G(\sigma)+m_-)}
      \theta^{m_--r(\sigma)}.
\end{align*}
An optimal $\sigma$ contributes at least $Z_0A^{s(K+m_-)}$ because $\theta\geq1$.
On the other hand, every $\sigma$ contributes at most $Z_0A^{s(K+m_-)}2^{m_-}$ because $\mathrm{score}_G(\sigma)\leq K$, $\theta\leq2$, and $0\leq m_--r(\sigma)\leq m_-$.
Summing over the $q^n$ choices of $\sigma$ shows that the total number $Z$ of proper colourings of $H$ satisfies
\begin{align*}
    Z_0 A^{(K + m_-)s}
    \leq Z
    \leq Z_0 A^{(K + m_-)s}q^n2^{m_-}.
\end{align*}

Suppose that an FPRAS for counting proper $q$-colourings exists. Run it with
a fixed relative-error parameter, say, $\eps=1/2$. With probability
at least $3/4$, its output $\widehat{Z}$ satisfies
\begin{align*}
    \underbrace{\frac{1}{2} Z_0 A^{(K + m_-)s}}_{:=L_K}
    \leq \widehat{Z}
    \leq \underbrace{2Z_0A^{(K + m_-)s}q^n2^{m_-}}_{:=R_K}.
\end{align*}
Since $G$ is planar, $m_-=O(n)$, and hence the second condition in the choice of $s$ gives
\begin{align*}
    \frac{R_K}{L_{K+1}}
    =\frac{4q^n2^{m_-}}{A^s}
    <1.
\end{align*}
Thus, $R_K<L_{K+1}$.
We can therefore obtain the exact value of $K$ by calculating the largest integer $x$ such that $2\widehat{Z}\geq Z_0A^{(x+m_-)s}$.
This completes the proof.
\end{proof}

\section{Two-spin system hardness}

In this section we prove Item~\eqref{item:2-spin-hard} of \Cref{thm:2-spin}, which states that computing the partition function of anti-ferromagnetic two-spin systems on planar graphs is \NP{}-hard.
Let $0 \le \beta,\gamma < 1$ be edge activities with $\max \set{\beta,\gamma}> 0$, and let $\lambda > 0$ be an external field. We assume $\beta \neq \gamma$ or $\lambda \neq 1$.
The starting point of our reduction is the classic \NP{}-hard problem~\cite{GJS76} of finding the maximum independent set in planar graphs.   
\prob{Planar-MIS}{
    A planar graph $G=(V,E)$;
}{
    The size of the maximum independent set in $G$.
}

\begin{proposition}[\cite{GJS76}]
    The problem \textsc{Planar-MIS} is \NP{}-hard.
\end{proposition}

We first introduce two rooted planar gadgets $U_0$ and $U_1$ that bias their roots toward spins $0$ and $1$, respectively, relative to the external field.

\begin{lemma}\label{lem:construct-gadget}
There exist constant-size rooted planar graphs $(U_0,r_0)$ and $(U_1,r_1)$, each admitting a plane embedding in which its root lies on the outer face, such that
\begin{align*}
    Z^{r_0 \gets 0}_{U_0}>Z^{r_0 \gets 1}_{U_0}
    \qquad\text{and}\qquad
    Z^{r_1 \gets 0}_{U_1}<Z^{r_1 \gets 1}_{U_1},
\end{align*}
where $Z^{r \gets a}_{U_b}$ denotes the partition function on $U_b$ given $r$ is pinned with $a\in\{0,1\}$ with the activity of the root omitted.
\end{lemma}

\begin{proof}
    We first construct a rooted planar graph $(W,r)$ with $Z_{W}^{r \gets 0} \neq Z_W^{r \gets 1}$. Consider the following constructions:
    \begin{itemize}
        \item $W=K_2$. In this case, $Z_W^{r \gets 0} = Z_W^{r \gets 1}$ only if
        \begin{align}\label{eq:bad-1}
            1 + \lambda \gamma = \beta + \lambda.
        \end{align}
        \item $W=K_3$. In this case, $Z_W^{r \gets 0} = Z_W^{r \gets 1}$ only if
        \begin{align}\label{eq:bad-2}
            \beta^3 + 2\beta\lambda+\gamma \lambda^2 = \beta + 2\gamma \lambda + \gamma^3 \lambda^2.
        \end{align}
    \end{itemize}
    Suppose both~\eqref{eq:bad-1} and~\eqref{eq:bad-2} hold. Plugging~\eqref{eq:bad-1} into~\eqref{eq:bad-2}, we have
    \begin{align*}
        0 &= (\beta^3-\beta) + \tp{2\beta - 2\gamma} \frac{1-\beta}{1-\gamma} + \tp{\gamma-\gamma^3} \tp{\frac{1-\beta}{1-\gamma}}^2\\
        &= \frac{1-\beta}{1-\gamma} \tp{-\beta(1+\beta)(1-\gamma) +2 (\beta-\gamma) + \gamma(1+\gamma)(1-\beta)}\\
        &= (\beta-\gamma)(1-\beta)^2.
    \end{align*}
    Combining with~\eqref{eq:bad-1}, it implies $\beta =\gamma$ and $\lambda = 1$, which violates our assumption. Thus, at least one of the constructions is a candidate with biased partition function. We construct $\widetilde{W}$ by taking $t$ copies of $W$, identifying all their roots into a single vertex $u$, and then adding a new vertex $\widetilde{r}$ adjacent to $u$. We designate $\widetilde{r}$ as the root of $\widetilde{W}$. 
    The parameter $t$ will be determined later.

    By a straightforward calculation,
    \begin{align}\label{eq:subtraction}
        Z_{\widetilde{W}}^{\widetilde{r} \gets 0} - Z_{\widetilde{W}}^{\widetilde{r} \gets 1} = \tp{Z_W^{r \gets 1}}^t \tp{(\beta-1) \tp{\frac{Z_W^{r \gets 0}}{Z_W^{r \gets 1}}}^t + \lambda (1-\gamma)}.
    \end{align}
    Therefore, for sufficiently large $t$, the quantity $Z_{\widetilde{W}}^{\widetilde{r} \gets 0}-Z_{\widetilde{W}}^{\widetilde{r} \gets 1}$ has the opposite sign to $Z_W^{r \gets 0}-Z_W^{r \gets 1}$. This completes the proof.
\end{proof}

We now present the reduction. Let $G=(V,E)$ be a planar graph. For every edge $uv\in E$, we replace the edge $uv$ by $s$ internally vertex-disjoint paths of length $3$ connecting $u$ and $v$. Afterwards, for every vertex $v\in V$, attach $x_v$ copies of $U_0$ and $y_v$ copies of $U_1$ by identifying the root of each copy with $v$. Let $H=(V',E')$ denote the resulting planar graph. Parameters $s, x_v,y_v$ will be specified in the proof.

\begin{proof}[Proof of Item~\eqref{item:2-spin-hard} of \Cref{thm:2-spin}]
After normalisation, the effective field $\lambda_v$ on $v\in V$ is
\begin{align*}
    \lambda_v \defeq \lambda \tp{\frac{Z^{r_0 \gets 1}_{U_0}}{Z^{r_0 \gets 0}_{U_0}}}^{x_v} \tp{\frac{Z^{r_1 \gets 1}_{U_1}}{Z^{r_1 \gets 0}_{U_1}}}^{y_v}.
\end{align*}
The weight contributed by the bundle replacing an edge $uv$ is $\zeta_{a,b}^s$, where $\zeta_{a,b}$ denotes the partition function of a path of length $3$ given that $u$ and $v$ are assigned spins $a$ and $b$, respectively, with the vertex weights of $u$ and $v$ omitted.
Note that $\zeta_{0,1}=\zeta_{1,0}$.
Effectively, the edge weights are $\beta'=\left(\frac{\zeta_{0,0}}{\zeta_{1,0}}\right)^s$ and $\gamma'=\left(\frac{\zeta_{1,1}}{\zeta_{1,0}}\right)^s$ for a $2$-spin system on $G$. 
Moreover, we can redistribute the $0,0$ edge weight $\beta'$ to the vertices, so that the effective field on $v$ becomes
\begin{align*}
    \lambda_v' \defeq \lambda \tp{\frac{Z^{r_0 \gets 1}_{U_0}}{Z^{r_0 \gets 0}_{U_0}}}^{x_v} \tp{\frac{Z^{r_1 \gets 1}_{U_1}}{Z^{r_1 \gets 0}_{U_1}}}^{y_v} \tp{\frac{\zeta_{1,0}}{\zeta_{0,0}}}^{s d_v},
\end{align*}
where $d_v$ is the degree of $v$ in $G$, and the effective edge weights become $\beta''=1$ and $\gamma''=\left(\frac{\zeta_{1,1}\zeta_{0,0}}{\zeta_{1,0}^2}\right)^s$.
This defines a new $2$-spin system on $G$ with edge weights $\beta''$ and $\gamma''$, and vertex weights $\lambda_v'$ for $v\in V$.
The partition function of the new spin system is equal to the partition function of the original spin system on $H$ up to a multiplicative factor that can be computed in linear time.
We next show that the new partition function is hard to approximate.


By~\Cref{lem:construct-gadget}, we have
\begin{align*}
    \frac{Z^{r_0 \gets 1}_{U_0}}{Z^{r_0 \gets 0}_{U_0}} < 1 
    \quad\text{ and }\quad \frac{Z^{r_1 \gets 1}_{U_1}}{Z^{r_1 \gets 0}_{U_1}} > 1.
\end{align*}
Since the original $2$-spin system is anti-ferromagnetic, $\zeta_{1,1} \zeta_{0,0} < \zeta_{1,0}^2$.
Let $n\defeq \abs{V}$.
Choose $s = \Theta(n^4)$ such that $\gamma''=\tp{\frac{\zeta_{1,1} \zeta_{0,0}}{\zeta^2_{1,0}}}^s \le \exp\tp{-n^4}$ and polynomially bounded parameters $x_v$, $y_v$, such that
\begin{align*}
\exp\tp{n^2} \le \lambda_v' \le U \exp\tp{n^2},
\end{align*}
where $U = \max \set{\frac{Z^{r_0 \gets 0}_{U_0}}{Z^{r_0 \gets 1}_{U_0}}, \frac{Z^{r_1 \gets 1}_{U_1}}{Z^{r_1 \gets 0}_{U_1}}}$. The existence of $s$, $x_v$ and $y_v$ can be verified easily.
Intuitively, the exponentially small $(1,1)$ edge weight ensures that a typical configuration is an independent set, and the exponentially large vertex weights ensure that contributions from the maximum independent sets dominate the partition function.
The reason that we have the $x_v$ part of the gadget is because the vertex weight might be too big to start with,
and for the reduction below we also need an upper bound on $\lambda_v'$.

For a configuration $\sigma\in\{0,1\}^{V}$ on $G$, let $S = \{v\in V \mid \sigma_v=1\}$ and let $w(\sigma)$ be the effective weight of $\sigma$ in the new spin system on $G$. 
If $S$ is not an independent set in $G$, it holds that for sufficiently large $n$,
\begin{align*}
    w(\sigma)\le \gamma''\cdot U^{\abs{S}} \cdot \exp\tp{n^2 \abs{S}} \le 1.
\end{align*}
Otherwise, 
\begin{align*}
\exp\tp{n^2 \abs{S}}\le w(\sigma)\le  U^{\abs{S}} \cdot \exp\tp{n^2 \abs{S}}.
\end{align*}
Let $K$ be the size of a maximum independent set in $G$. 
Therefore, the partition function $Z_{\textrm{new}}=\sum_{\sigma \in \{0,1\}^V} w(\sigma)$ satisfies
\begin{align*}
    \exp\tp{n^2 K}\le Z_{\textrm{new}} \le (2 U)^n \exp\tp{n^2 K} + 2^n.
\end{align*}

Suppose that an FPRAS for estimating the partition function under the given parameters exists. Run it with
a fixed relative-error parameter, for example $\eps=1/2$. With probability
at least $3/4$, its output $\widehat{Z}$ satisfies
\begin{align*}
    \underbrace{\frac{1}{2}\exp\tp{n^2 K}}_{:=L_{K}} \le \widehat{Z} \le \underbrace{(4 U)^n \exp\tp{n^2 K} + 2^{n+1}}_{:=R_{K}}.
\end{align*}
When $n$ is sufficiently large, $L_{K+1} > R_K$ holds and thus the size of a maximum independent set can be obtained by calculating the largest integer $\ell$ with $\frac{1}{2}\exp\tp{n^2 \ell} \le \widehat{Z}$.
Thus, the assumed FPRAS yields a randomised polynomial-time algorithm
that computes the maximum independent-set size exactly with probability
at least $3/4$. This gives an RP
algorithm for \textsc{Planar-MIS} and hence implies $\NP{}=\RP{}$. This completes the proof.
\end{proof}

\section*{Acknowledgement}

HG would like to thanks Mark Jerrum for extensive discussions on the planar hardcore problem.
We also thank Zongchen Chen for pointing out connections between our work and the pairwise influence matrix introduced in~\cite{LO25}.

\bibliographystyle{alpha}
\bibliography{ref}

\end{document}